\documentclass[onecollarge]{svjour2}                    
\smartqed  
\usepackage{graphicx}
\usepackage{amsmath,amsfonts,amssymb,amsopn,amscd}
\usepackage{bbm,wasysym,stmaryrd}
\usepackage{subfigure}
\usepackage{color,colortbl}
\usepackage{hyperref}

\newcommand{\R}{\mathbbm{R}}
\newcommand{\Z}{\mathcal{Z}}

\renewcommand{\P}{\mathbbm{P}}
\newcommand{\Pc}{\mathcal{P}}

\newcommand{\N}{\mathbbm{N}}
\newcommand{\F}{\mathcal{F}}
\newcommand{\B}{\mathcal{B}}
\newcommand{\V}{\mathcal{V}}
\newcommand{\E}{\mathcal{E}}
\newcommand{\A}{\mathcal{A}}

\newcommand{\norm}[2]{ \left \Vert #2 \right\Vert_{\mathcal{Z}_{#1}}}

\usepackage{enumitem}

\DeclareMathOperator{\eqlaw}{\stackrel{\mathcal{L}}{=}} 

\newcommand{\ind}[1]{\mathbbm{1}_{#1}}
\newcommand{\Exp}{\mathbbm{E}} 
\newcommand{\Ex}[1]{\mathbbm{E}\left [#1 \right]} 

\newcommand{\M}{\mathcal{M}}

\newcommand{\Ct}{\mathcal{C}_{\tau}}

\journalname{Journal of Statistical Physics}

\graphicspath{{./Figures/}}

\begin{document}

\title{Spatially extended networks with singular multi-scale connectivity patterns\thanks{INRIA BANG Laboratory and the Mathematical Neuroscience Lab, CIRB-Collège de France, CNRS UMR 7241m INSERM U1050, Universit\'e Pierre et Marie Curie ED 158. MEMOLIFE Laboratory of excellence and Paris Sciences Lettres PSL*.}
}

\author{Jonathan Touboul}

\institute{ The Mathematical Neuroscience Laboratory \\
			Coll\`ege de France / CIRB and INRIA Bang Laboratory\\
			11, place Marcelin Berthelot,
			75005 Paris, France
            Tel.: +33-144271388\\
            \email{jonathan.touboul@college-de-france.fr}           
}

\date{Received: date / Accepted: date}

\maketitle

\begin{abstract}
The cortex is a very large network characterized by a complex connectivity including at least two scales: a microscopic scale at which the interconnections are non-specific and very dense, while macroscopic connectivity patterns connecting different regions of the brain at larger scale are extremely sparse. This motivates to analyze the behavior of networks with multiscale coupling, in which a neuron is connected to its $v(N)$ nearest-neighbors where $v(N)=o(N)$, and in which the probability of macroscopic connection between two neurons vanishes. These are called singular multi-scale connectivity patterns. We introduce a class of such networks and derive their continuum limit. We show convergence in law and propagation of chaos in the thermodynamic limit. The limit equation obtained is an intricate non-local McKean-Vlasov equation with delays which is universal with respect to the type of micro-circuits and macro-circuits involved. 
\end{abstract}

\bigskip

\hrule
\setcounter{tocdepth}{3}

\tableofcontents

\medskip

\hrule

\medskip

The purpose of this paper is to provide a general convergence and propagation of chaos result for large, spatially extended networks of coupled diffusions with multi-scale disordered connectivity. Such networks arise in the analysis of neuronal networks of the brain. Indeed, the brain cortical tissue is a large, spatially extended network whose dynamics is the result of a complex interplay of different cells, in particular neurons, electrical cells with stochastic behaviors. In the cortex, neurons interact depending on their anatomical locations and on the feature they code for. The neuronal tissue of the brain constitute spatially-extended structures presenting complex structures with local, dense and non-specific interactions (microcircuits) and long-distance lateral connectivity that are function-specific. In other words, a given cell in the cortex sends its projections at (i) a local scale: the neurons connect extensively to anatomically close cells (the \emph{microcircuits}), forming a dense local network, and (ii) superimposed to this local architecture, a very sparse functional architecture arises, in which long-range connections are made with other cells that are anatomically more remote but that respond to the same stimulus (the functional \emph{macrocircuit}). This canonical architecture was first evidenced by electrophysiological recordings in the 70's~\cite{hubel-wiesel:77,mountcastle:57} and made more precise as experimental techniques developed (see~\cite{bosking-zhang-etal:97} for striking representations of this architecture in the striate cortex). The primary visual cortex  of certain mammals is a paradigmatic well documented cortical area in which this architecture was evidenced. In such cortical areas, neurons organize into columns of small spatial extension containing a large number of cells (on the order of tens of thousands cells) responding preferentially to specific orientations in visual stimuli~\cite{hubel-wiesel-etal:78}, constituting local microcircuits that distribute across the cortex in a continuous map, each cell connecting densely with its nearest neighbors and sparsely with remote cells coding for the same stimulus~\cite{bosking-zhang-etal:97}. These spatially extended networks are called \emph{neural fields}.

Such organizations and structures are deemed to subtend processing of complex sensory or cortical information and support brain functions~\cite{kandel-schwartz-etal:00}. In particular, the activity of these neuronal assemblies produce a mesoscopic, spatially extended signal, which is precisely at the spatial resolution of the most prominent imaging techniques (EEG, MEG, MRI). These recordings are good indicators of brain activity: they are a central diagnostic tool used by physicians to assert function or disfunction. 

In these spatially extended systems, the presence of delays in the communication of cells, chiefly due to the transport of information through axons and to the typical time the synaptic machinery needs to transmit it, is essential to the dynamics. These transmission delays will chiefly affect the long connections of the macrocircuit, which are orders of magnitude longer than those of the microcircuit. 

The mathematical and computational analysis of the dynamics of neural fields relies almost exclusively on the use of heuristic models since the seminal work of Wilson, Cowan and Amari \cite{amari:72,wilson-cowan:73}. These propose to describe the mesoscopic cortical activity through a deterministic, scalar variable whose dynamics is given by integro-differential equations. This model was widely studied analytically and numerically, and successfully accounted for hallucination patterns, binocular rivalry and synchronization~\cite{laing-troy-etal:02,ermentrout-cowan:79}. Justifying these models starting from biologically realistic settings has since then been a great endeavor~\cite{bressloff:12}. 

This problem was undertaken recently using probabilistic methods. The first contribution~\cite{touboulNeuralfields:11} introduced an approximation of the underlying connectivity of the neural network involved, considering a fully connected architecture (each neuron was connected to all the others) and neurons in the same column were considered to be precisely at the same spatial location. They showed propagation of chaos and convergence to some intricate McKean-Vlasov equation. More recently, an heterogeneous macrocircuit model was analyzed in~\cite{stannat-lucon:13}. In that paper, the authors considered a network with heterogeneous and non-global connectivity: neurons were connected with their $P$-nearest neighbors, where $P=cN$ with $c<1$, or with power-law synaptic weights, and obtained a limit theorem for the behavior of the empirical density. In both cases, the connectivity was considered at a single scale, and did not reproduce the actual type of connectivity pattern observed in the brain. 

In the present manuscript, we come back to these models with a more plausible architecture including local microcircuit together with non-local macroscopic sparse connectivity. Using statistical methods and in particular an extension of the coupling method~\cite{mckean:66,dobrushin:70,sznitman:89}, we will demonstrate the propagation of chaos property, and convergence towards a complex nonlinear Markov equation similar to the classical \emph{McKean-Vlasov} equations, but with a non-local integral over space locations and delays. Interestingly, this object presents substantial differences with the usual McKean-Vlasov limits: beyond the presence of delays, the neural field limit regime is at a mesoscopic scale where averaging effects locally to occur, but is fine enough to resolve brain's structure and its activity, resulting in the presence of an integral term over space. The solution, seen as a function of space, is everywhere discontinuous, which makes the limiting object highly singular. The present work {is distinct of that of~\cite{stannat-lucon:13} in that we consider local connectivity patterns in which neurons connect to a negligible portion of the neurons}. This  includes non-trivial issues, and necessitate to thoroughly control the regularity of the law of the solution as a function of space. On the other hand, beyond the presence of random locations of individual neurons and the presence of a dense microcircuit, the sparse macro-circuit generalizes non-trivially the work done in~\cite{touboulNeuralfields:11}. Indeed, at the macro-circuit scale, the probability of connecting two fixed neurons tends to zero. We therefore need to deal with a non-globally connected network, and address the problem by using fine estimates on the interaction terms and Chernoff-Hoeffding theorem~\cite{hoeffding1963probability}. 

The speed of convergence towards the mean-field equations is quantified and involves three terms, one governing the local averaging effects arising from the micro-circuits, one arising from the regularity properties of the solutions, and one corresponding to the speed of convergence of the macro-circuit interaction term towards a continuous limit. In the neural field regime, the limit equations are very singular, in particular trajectories are not measurable with respect to the space. These limits are very hard to analyze at this level of generality. However, in the type of models usually considered in the study of neural fields, namely the firing-rate model, it was shown in~\cite{touboulNeuralFieldsDynamics:11} that the behavior can be rigorously and exactly reduced to a system of deterministic integro-differential equations that are compatible with the usual Wilson and Cowan system in the zero noise limit. Noise intervenes in these equations a nonlinear fashion, fundamentally shaping in the macroscopic dynamics. 

	The paper is organized as follows. We start in section~\ref{sec:model} by introducing precisely our model and proving a few simple results on the network equations and on the topology of the micro-circuit. This being shown, we will turn in section~\ref{sec:ExistenceUniquenessSpace} to the analysis of the network equations, and will in particular make sense of the intricate non-local McKean-Vlasov equation, show well-posedness and some regularity estimates on the law of the mean-field equations. Section~\ref{sec:PropaChaSpace} will be devoted to the demonstration of the convergence of the network equations towards the mean-field equations. 
	
	\section{Model and Network Equations}\label{sec:model}
	
We consider a piece of cortex $\Gamma$ (the \emph{neural field}), which is a regular compact subset\footnote{It is generally chosen to be a smooth open set such as a rectangle in $\R^2$ when representing locations on the cortex, or periodic domains such as the torus of dimension 1 $\mathbbm{S}^1$ in the case of the representation of the visual field, in which neurons code for a specific orientation in the visual stimulus: in that model, $\Gamma$ is considered to be the feature space~\cite{bressloff-cowan-etal:01,ermentrout-cowan:79}.} of $\R^d$ for some $d\in \N^*$, and the density of neurons on $\Gamma$ is given by a probability measure $\lambda\in \M^{1}(\Gamma)$ assumed to be absolutely continuous with respect to Lebesgue's measure $dr$ on $\Gamma$, with strictly positive and bounded density $d\lambda/dr \in [\lambda_{min},\lambda_{max}]$.
	
	 On $\Gamma$, we consider a spatially extended network composed of $N$ neurons at random locations $r_i\in\Gamma$ drawn independently with law $\lambda$ in a probability space $(\Omega',\F',\Pc)$, and we will denote by $\E$ the expectation with respect to this probability space. A given neuron $i\in\{1,\cdots,N\}$ projects local connections in its neighborhood $\V(i)$, and long-range connections over the whole neural field. We will consider here that the local microcircuit connectivity consists of a fully connected graph with $v(N)=o(N)$ nearest-neighbors\footnote{Random connections can be introduced and handled with the same techniques as done for the macrocircuit. }. The synaptic weights corresponding to these connections are assumed equal to 
	\[\bar{J}/v(N),\] 
	where $\bar{J}\in\R$ (it is generally positive since local interactions in the cortex tend to be excitatory). A central example is the case $v(N)=cN^{\alpha}$ with $\alpha<1$. 
	\begin{remark}
		With zero probability, it may occur for a given neuron $i$ that its local microcircuit $\V(i)$ is not well defined. This occurs if there exists $R>0$ such that the number of neurons at distance strictly smaller than $R$ of neuron $i$, denoted $v_R(i,N)$, is strictly smaller than $v(N)$ and the number of neurons at a distance smaller or equal to $R$ is strictly larger than $v(N)$, meaning in particular that there exists several neurons at distance precisely $R$. This event has a null probability, $\V(i)$ is defined as the union of all neurons at distance strictly smaller than $R$, completed by $v(N)-v_R(i,N)$ neurons randomly chosen among those at distance exactly $R$ of neuron $i$. 
	\end{remark}
	The neurons also send non-local connections which are specific (i.e. depend on the type of neurons, indexed here by the spatial location), which are much sparser than the local microcircuit. We will consider that the macro-connections are random variables $M_{ij}$ drawn in $(\Omega',\F',\Pc)$ and frozen during the evolution of the network, with law:
	\[M_{ij}=\frac{J(r_i,r_j)}{N\beta(N)} \chi_{ij}\]
	where $\chi_{ij}$ is a Bernoulli random variable with parameter $\beta(N)$
	\[\chi_{ij}=\begin{cases}
		 1 & \text{with probability } \beta(N)\\
		 0 & \text{with probability } 1-\beta(N).
	\end{cases}\]
	The coefficient $J(r_i,r_j)$ governs the connectivity weight between neurons at location $r_i$ and $r_j$. For instance, in the visual cortex, if the neurons of the cortical column at location $r_i$ codes for the collinear (resp, orthogonal) orientation as neurons in the column at $r_j$, $J(r_i,r_j)$ is positive (negative). These coefficients are assumed to be smooth (see assumption~\ref{Assump:SpaceContinuity}) and bounded, and we denote:
	\[\Vert J\Vert_{\infty}=\sup_{(r,r')\in\Gamma^2}\vert J(r,r')\vert. \]
	The scaling coefficient $\beta(N)v(N)$ corresponds to the total incoming connections from the microcircuit related to neuron $j$. The parameter $\beta(N)$ accounts for the connectivity level of the macrocircuit. In particular, if populations are not connected, we will set $J(r_i,r_j)=0$. In that sense, the function $\beta(N)$ does not account for all absent links in the network, but rather for the sparsity of the macro-circuit. Motivated by the fact that the macro-circuit is very sparse and that micro-circuits form non-trivial patches of connectivity, we will assume that, when $N\to \infty$,
	\[\begin{cases}
		\beta(N) \to 0\\
		\beta(N)v(N) \to \infty.
	\end{cases}\]
	The hypothesis on the connectivity ensure the following facts, desirable for a modeling at the neural field scale (see Fig.~\ref{fig:Neurons}):
	\begin{itemize}
		\item the local micro-circuit shrinks to a single point in the limit $N\to \infty$ (see lemma~\ref{lem:SizeMicro}), and
		\item the macro-circuit is sparse at the level of single cells ($\beta(N)\to 0$), but non-sparse at the level of cortical columns ($\beta(N)v(N)\to\infty$). 
	\end{itemize}
	Note that in all our developments, one only needs the assumption that $N\beta(N)\to\infty$ as $N\to\infty$. This is of course a consequence of our current assumption. 
	
	\begin{remark} \ 
A schematic topology usually considered could be the 2-dimensional regular lattice 
\[\Gamma^N=\left\{\left(\frac j N ; \frac i N\right) ; 1\leq i,j \leq N\right\}\]
approximating the unit square $\Gamma=[0,1]^2$ with $N^2$ points. In this model, typical micro-circuit size could be chosen to be $cN^{\alpha}$ with $\alpha<1$, and $\beta(N)$ of order $N^{-\gamma}$ with $0<\gamma<\alpha$. Our model takes into account the fact that in reality, neurons are not regularly placed on the cortex, and therefore such a regular lattice case is extremely unlikely to arise (this architecture has probability zero). Moreover, in contrast with this more artificial example, the probability distribution of the location of one given neuron do not depend on the network size. In our setting, $\lambda$ accounts for the density of neurons on the cortex, and as the network size is increased, new neurons are added on the neural field at locations independent of that of other neurons, with the same probability $\lambda$, so that neuron locations sample the asymptotic cell density.   
	\end{remark}
	
	These elements describe the random topology of the network. Prior to the evolution, a number of neurons $N$ and a configuration $\A_N$ is drawn in the probability space $(\Omega',\F',\Pc)$. The configuration of the network provides:
	\begin{itemize}
		\item The locations of the neurons $(r_i)_{i=1\cdots N}$ i.i.d. with law $\lambda$
		\item The connectivity weights, in particular the values of the i.i.d. Bernoulli variables $(\chi_{ij})_{i,j\in 1\cdots N}$ of parameter $\beta(N)$.
	\end{itemize}

	\begin{figure}[htbp]
		\centering
			\includegraphics[width=.5\textwidth]{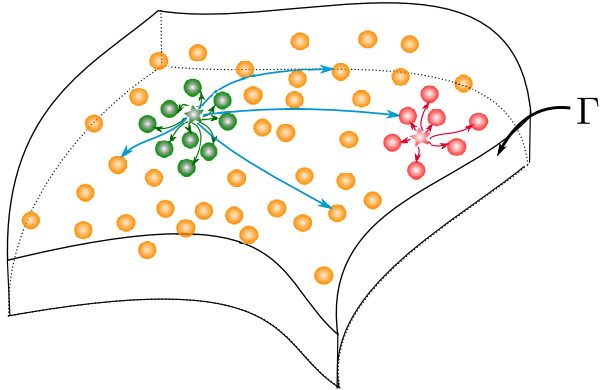}
		\caption{A typical multi-scale architecture of neural fields. Green (red) circles correspond to neurons belonging to the local microcircuits of the green (red) star neuron, and blue arrows correspond to macro-connections of the green star neuron. }
		\label{fig:Neurons}
	\end{figure}
		
	Let us start by analyzing the topology of the micro-circuit. At the macroscopic scale, we expect local micro-circuits to shrink to a single point in space, which would precisely correspond to the scale at which imaging techniques record the activity of the brain (a pixel in the image). The micro-circuit connects a neuron to its $v(N)$ nearest neighbors. We made the assumptions that $v(N)$ tends to infinity as $N\to \infty$ while keeping $v(N)=o(N)$. This property ensures that for a fixed neuron $i\in\{1\cdots N\}$ and for any $r_{j}\in \V(i)$, the distance\footnote{In all the manuscript, we will denote, for any $p\in\N^*$ and $x\in \R^p$, $\vert x\vert$ the Euclidean norm of $x$, regardless of the space involved and the dimension $p$ considered.} $d_{ij}=\vert r_j - r_i\vert$ is, with overwhelming probability, upperbounded by a constant multiplied by $(v(N)/N)^{1/d}$. In the regular lattice case, this property is trivial. In our random setting, we introduce the maximal distance between two neurons in the microcircuit associated to neuron $i$ is noted:
	\[d_m(i)=\max_{j\in\V(i)} d_{ij}.\] 
	This quantity has a law that is independent of the specific neuron $i$ chosen. 
	\begin{proposition}\label{lem:SizeMicro}
		The microcircuit shrinks to a single point in space as $N$ to infinity. More precisely, for any $i\in\{1,\cdots,N\}$, the maximal distance $d^N_m(i)$ between two neurons in the microcircuit associated to neuron $i$ decreases towards $0$, in the sense that there exists $C>0$ such that the maximal distance between two points in a microcircuit satisfies the inequality:
		 \[\max_{i\in\{1\cdots N\}}\mathcal{E} [d^N_m(i)] \leq C \left(\left(\frac{v(N)}{N}\right)^{\frac 1 d} + \frac{1}{v(N)}\right).\]
	\end{proposition}
	
	 \begin{proof}
	 	We have assumed that the locations $r_i$ are iid with law $\lambda$ absolutely continuous with respect to Lebesgue's measure with density lowerbounded by some positive quantity $\lambda_{min}>0$. Let us fix a neuron $i$ at location $r_i$ which is almost surely in the interior of $\Gamma$. We are interested in the distances between different neurons within the microcircuit around $i$, and will therefore consider the distribution of relative locations of neurons belonging to $\V(i)$ conditionally to the location $r_i$ of neuron $i$. We will denote $\E_i$ the expectation under this conditioning. It is clear that the set of random variables $\{d_{ij}\}$ are identically distributed. Moreover, these are independent conditionally on the value of $r_i$. We will show that, for any neuron $j\in\V(i)$, the distance $d_{ij}$ tends to zero as $N$ increases with probability one. To this end, we use the characterization of the maximal distance $d_m(i)=\max_{j\in\V(i)} d_{ij}$ as the minimal radius $r$ such that the ball centered at $r_i$ with diameter $r$ contains $v(N)$ points: 
	\[d_{m}(i)\leq \inf\left\{r>0; M_N(r):= \sum_{j=1}^N \ind{B(r_i,r)} (r_j) \geq v(N)\right\}.\]
	We will show that there exists a quantity $\alpha(N)$ tending to zero such that $d_m(i)\leq \alpha(N)$ with large probability, i.e. $M_N(\alpha(N))>v(N)$. To this end, we start by noting that conditionally on $r_i$, the random variables $(z_j^N=\ind{B(r_i,\alpha(N))} (r_j) ; j=\{1\cdots N\})$ are independent, identically distributed. Moreover, for $r_i$ fixed, there exists $N_0>0$ such that for any $N>N_0$, $B(r_i,\alpha(N))\subset \Gamma$. Therefore, for $N>N_0$, the random variables $z_j^N$ are such that:
	\[\begin{cases}
		\E_i{[z_j^N]}=\int_{B(r_i,\alpha(N))} d\lambda(r)\in [\lambda_{min},\lambda_{max}]\times \gamma(N)\\
		\text{Var}_i(z_j^N) = \E_i{[z_j^N]} - \E_i{[z_j^N]}^2 \in [\lambda_{\min} - \lambda_{\max} \gamma(N),\lambda_{\max}-\lambda_{\min} \gamma(N)] \times \gamma(N)
	\end{cases}\]
	where $[a,b]\times c = [ac,bc]$ and 
	\[\gamma(N)=K_d \alpha(N)^d.\] 
	with $K_d=\frac{\pi^{d/2}}{\Gamma(\frac d 2 + 1)}$ the volume of the unit ball in $\R^d$. 
	The radius $\alpha(N)$ is chosen such that:
	\[\alpha(N)= \eta \left(\frac{1}{K_d\lambda_{min}}\frac{v(N)}{N}\right)^{1/d}\]
	with $\eta>1$. This assumption implies that 
	\[\gamma(N) = \frac{\eta^d}{\lambda_{\min}} \frac{v(N)}{N} \qquad \text{ and } \qquad \E_i[z_1^N]\geq \eta^d \frac{v(N)}{N}\]
	We have:
	\[\frac 1 {v(N)}\sum_{j=1}^N z_j^N - \frac {N}{v(N)}\E_i{[z_j^N]}= \frac 1 {v(N)}\sum_{j=1}^N (z_j^N-\E_{i}{[z_j^N]}) .\]
	which tends to zero in probability, since (we recall that the variance is of order $\gamma(N)$ for $N$ large)
	\begin{align*}
		\E_i\left[ \left(\frac 1 {v(N)}\sum_{j=1}^N \left(z_j^N-\E_i{[z_j^N]}\right)\right)^2\right] = \frac N {v^2(N)} \text{Var}_i(z_j^N) = O \left(\frac 1 {v(N)}\right).
	\end{align*}
	The quantity $M_N(\alpha(N))/v(N)$ is therefore lowerbounded, with overwhelming probability, by $ \lambda_{\min} \frac{1}{K_d}\frac N {v(N)} \alpha(N)^d$ which is, under our assumption on $\eta$, is greater than $1$ for $N$ large enough: with overwhelming probability, for large $N$, the microcircuit is fully included in the ball of radius $\alpha(N)$. 
	
	We have assumed that the set $\Gamma$ is bounded. Let us denote by $d(\Gamma)$ its diameter (i.e. the maximal distance between two points in $\Gamma$). We have:
	\begin{align*}
		\mathcal{E}_i [d^N_m(i)] &= \E_i[d^N_m(i) \mathbbm{1}_{d^N_m(i)\leq 2\alpha(N)}] + \E_i[d^N_m(i) \mathbbm{1}_{d^N_m(i)> 2\alpha(N)}]\\
		& \leq 2\alpha(N) + d(\Gamma)\Pc_i\left[M_N(\alpha(N))<v(N) \right]\\
		& \leq 2\alpha(N) + d(\Gamma)\Pc_i\left[ \left \vert \frac{M_N(\alpha(N))}{v(N)}-\frac{N}{v(N)} \E_i{[z_j^N]} \right\vert  >\eta^d-1 \right]\\
		& \leq 2\alpha(N) + d(\Gamma)\frac{\E_i{ \left[\left(\frac 1 {v(N)}\sum_{j=1}^N \left(z_j^N-\E_i{[z_j^N]}\right)\right)^2\right]}}{(\eta^d-1)^2} \\
		&\leq C \left(\left(\frac{v(N)}{N})\right)^{\frac 1 d} + \frac{1}{v(N)}\right)
	\end{align*}
	which ends the proof. 
\end{proof}

	Let us now introduce the dynamics of the neurons activity. The state of each neuron $i$ in the network is described by a $q$-dimensional variable $X^i\in E:=\R^{q}$, typically corresponding to the membrane potential of the neuron and possibly additional variables such as those related to ionic concentrations and gated channels. These variables have a stochastic dynamics. In order to deal with these stochastic evolutions, we introduce a new complete probability space $(\Omega, \F,\P)$ endowed with a filtration $(\F_t)_{t\geq 0}$ satisfying the usual conditions, and we denote by $\Exp$ the expectation with respect to this probability space. Note that this space is distinct from the configuration space $(\Omega',\F',\Pc)$. Once a configuration $\A_N$ is fixed for a $N$-neurons network, it is frozen and each neuron will have a random evolution following the equations:
	\begin{multline}\label{eq:NetworkSpace}
		d\,X^{i,\A_N}_t=\Big(f(r_{i},t,X^{i,\A_N}_t)
		+ \frac{\bar{J}}{v(N)}\sum_{j\in \V(i)} b(X^{i,\A_N}_t,X^{j,\A_N}_{t-\tau(r_i,r_j)}) + \sum_{j=1}^N M_{ij} b(X^{i,\A_N}_t,X^{j,N}_{t-\tau(r_i,r_j)}) \Big) \,dt \\+ \sigma(r) d\tilde{W}^{i}_t.
	\end{multline}
	where $f(r,t,x):\Gamma\times\R\times E\mapsto E$ governs the intrinsic dynamics of each cell, $(\tilde{W}_t^i)$ is a sequence of independent $(\Omega,\F, (\F_t),\P)$ Brownian motions of dimension $m$ modeling the external noise, $\sigma(r):\Gamma\mapsto \R^{q\times m}$ a bounded and measurable function of $r\in\Gamma$ modeling the level of noise at each space location, and $b: E^2\mapsto E$ the interaction function. The map $\tau(r,r'):\Gamma^2\mapsto \R^+$ is the interaction delay between neurons located at $r$ and those at $r'$ which is assumed to be of the form:
	\[ \tau(r,r')=\tau_s+\vert r-r'\vert/c\] 
	where $\tau_s$ is the synaptic transmission time and $\vert r-r' \vert/c$ the transport time ($c$ is the transmission speed assumed constant). Since $\Gamma$ is bounded, all delays are bounded by a finite quantity $\tau$ (in our notations, $\tau=\tau_s+d(\Gamma)/c$). In what follows, we will use the shorthand notation $\tau_{ij}=\tau(r_i,r_j)$.
	The parameters of the system are assumed to satisfy the following assumptions:
	\renewcommand{\theenumi}{(H\arabic{enumi})}
	\begin{enumerate}
		\item\label{Assump:LocLipschSpace} $f$ is $K_f$-Lipschitz-continuous with respect to all three variables,  
		\item\label{Assump:LocLipschbSpace} $b$ is $L$-Lipschitz-continuous with respect to both variables and bounded. We denote
		\[\Vert b\Vert_{\infty}=\sup_{(x,y)\in E^2} \vert b(x,y)\vert\]
	\item\label{Assump:LinearGrowth} The drift satisfies uniformly in space ($r$) and time ($t$), the inequality:
	\[x^T f(r,t,x)+ \vert \sigma(r)\vert^2 \leq C \; (1+\vert x \vert^2)\]
	\item\label{Assump:SpaceContinuity} The drift, delay, diffusion and connectivity functions are regular with respect to the space variables $(r,r')\in \Gamma^2$ (we will assume for instance that these are all $K_{\Gamma}$-Lipschitz continuous).
	\end{enumerate}
	
	Let us first state the following proposition ensuring well-posedness of the network system under the assumptions of the section:
	\begin{proposition}\label{pro:ExistenceUniquenessNetwork}
		Let $(X^0_t)_{t\in[-\tau,0]}$ a square integrable process with values in $E^N$. Under the current assumptions, for any configuration $\A_N$ of the network, there exists a unique strong solution to the network equations \eqref{eq:NetworkSpace} with initial condition $X^0$. This solution is square integrable and defined for all times.
	\end{proposition}

	The proof of this proposition is classical. It is a direct application of the general theory of SDEs in infinite dimensions~\cite[Chapter 7]{da-prato:92}, and elementary proof in our particular case of delayed stochastic differential equations can be found in~\cite[Theorem 5.2.2]{mao:08}: for any fixed configuration, we have a regular $N$-dimensional SDE with delays satisfying a monotone growth condition~\ref{Assump:LinearGrowth} ensuring a.s. boundedness for all times of the solution. The proof of this property is essentially based on the same arguments as those of the proof of theorem~\ref{thm:ExistenceUniquenessSpace}, and the interested reader is invited to follow the steps of that demonstration. 
	
	It is important to note that the bound one obtains on the expectation of the squared process depends on the configuration of the network. Indeed, the macroscopic interaction term involves the sum of a random number of terms $n_{i}(\A_N) = \sum_{j}\chi_{ij}$ rescaled by $\frac{1}{N\beta(N)}$. The quantity $n_i$ can take large values (up to $N$) with positive (but small) probability, and therefore the scaling coefficient is not enough to properly control such cases. The bound obtained by classical methods will therefore diverge in $N$, and this will be a deep question for our aim to prove convergence results as $N\to\infty$. In the present manuscript, we will be able to handle these terms properly in that limit by using fine estimates related to $n_i(\A_N)$, see lemma~\ref{lem:SumChi}.

We are interested in the limit, as $N\to\infty$, of the behavior of the neurons. Since we are dealing with diffusions in random environment, there are at least two notions of convergence: \emph{quenched} convergence results valid for almost all configuration $\A$, and \emph{annealed} results valid for the law of the network averaged across all possible configurations. Here, we will show averaged convergence results as well as quenched properties along subsequences. 

Similarly to what was observed in~\cite{touboulNeuralfields:11}, the limit of such spatially extended mean-field models will be stochastic processes indexed by the space variable, which, as a function of space, are not measurable with respect to the Borel algebra $\B(\Gamma)$. As noted in~\cite{touboulNeuralfields:11}, this is not a mathematical artifact of the approach, since neurons accumulating on the neural field are driven by independent Brownian motions, and therefore no regularity is to be expected in the limit. However, even if trajectories are highly irregular, this will not be the case of the law of these solutions. In order to handle this irregularity, we will use the \emph{spatially chaotic} Brownian motion on $\Gamma$, a two-parameter process $(t,r)\in\R^+\times\Gamma \mapsto W_t(r)$ such that for any fixed $r\in\Gamma$, the process $t\mapsto W_t(r)$ is a $m$-dimensional standard Brownian motion, and for $r\neq r'$ in $\Gamma$, the processes $W_t(r)$ and $W_t(r')$ are independent\footnote{ We will also use the terminology of~\cite{touboulNeuralfields:11} and will qualify a process $\zeta_t(r)$ of \emph{spatially chaotic} if the processes $\zeta_t(r)$ and $\zeta_t(r')$ are independent for any $r\neq r'$.}. This process is relatively singular seen as a spatio-temporal process: in particular, it is not measurable with respect to $\mathcal{B}(\Gamma)$. The spatially chaotic Brownian motion is distinct from other more usual spatio-temporal processes. In particular, its covariance is $\Ex{W_t(r)W_{t'}(r')}=(t\wedge t') \delta_{r=r'}$: the covariance is not measurable with respect to $\B(\Gamma)$. In contrast, the more classical space-time Brownian motion (the process corresponding to space-time white noise differential terms) on the positive line ($\Gamma=\R^+$) has a covariance $(t\wedge t') (r\wedge r')$: it is continuous with respect to space. It is also distinct from Wiener processes on Hilbert spaces~\cite[Chapter 4.1.]{da-prato:92} (a.k.a. cylindrical Brownian motions) which have a covariance defined through a trace-class operator on the Hilbert space, and may be decomposed as the sum of standard Brownian motions on a basis of that Hilbert space (i.e., there is a countable number of Brownian motions involved). The chaotic Brownian motion, due to his high singularity as a space-time process, is more suitably seen as an infinite collection of Wiener processes. 

	We will show that the network equations~\eqref{eq:NetworkSpace} satisfies the propagation of chaos property in the limit where $N$ goes to infinity, and that the state of the network converges towards a very particular McKean-Vlasov equation involving a spatially chaotic Brownian motion. The propagation of chaos property (Boltzmann's molecular chaos hypothesis, or Sto\ss zahlansatz) states that, provided that the initial conditions of all neurons are independent, the law of any finite set of neurons converge to a product of laws (loosely speaking, are asymptotically independent) for all times (see~\cite{sznitman:89}). In our network, this property means that the heuristic notion of Boltzmann's Sto\ss zahlansatz applies in that the dependence relationship between the state of, say 2, fixed neurons in the network, dilute away in the thermodynamic limit, and these end up being independent in the large $N$ limit. In detail, for almost all configuration of the network, the asymptotic law of neurons located at $r$ in the support of $\lambda$ will be measurable with respect to $(\Gamma,\B(\Gamma))$ and converge towards the stochastic neural field mean-field equation with delays:
	\begin{multline}\label{eq:MFESpace}
		d\,\bar{X}_t(r)=f(r,t,\bar{X}_t(r)) \, dt + \sigma(r) dW_t(r) + \bar{J}\Exp_{\bar{Z}}[b(\bar{X}_t(r),\bar{Z}_{t-\tau_s}(r))]\,dt\\
		+ \int_{\Gamma} J(r,r')\Exp_{\bar{Z}}[b(\bar{X}_t(r),\bar{Z}_{t-\tau(r,r')}(r'))] \, d\lambda(r') \,dt 
	\end{multline}
	where $(W_t(r))_{t\geq 0, r\in\Gamma}$ is a spatially chaotic Brownian and the process $(\bar{Z})$ is independent and has the same law as $(\bar{X})$. In other words, we will show that the law of the solution $X_t(r)$, noted $m(t,r,dy)$, is measurable with respect to $\mathcal{B}(\Gamma)$, and that the mean-field equation can be expressed as the integro-differential McKean-Vlasov equation:
	\begin{multline*}
		d\,\bar{X}_t(r)=f(r,t,\bar{X}_t(r)) \, dt 
		+ \sigma(r) dW_t(r) +\bar{J}\int_E b(\bar{X}_t(r),y ) m(t-\tau_s,r,dy) \,dt \\
		+ \int_{\Gamma} J(r,r') \int_E b(\bar{X}_t(r),y ) m(t-\tau(r,r'),r',dy) \, d\lambda(r') \,dt .
	\end{multline*}
	Let us eventually give the Fokker-Planck equation on the possible density $p(t,r,y)$ of $m(r,t)$ with respect to Lebesgue's measure:
	\begin{multline}\label{eq:FP}
		\partial_t p(t,r,x)=-\nabla_x\Bigg \{\bigg(f(r,t,x)  + \bar{J}\int_E b(x,y) p(t-\tau_s,r,y)\,dy \\+ \int_{\Gamma} J(r,r')\int_E b(x,y ) p(t-\tau(r,r'),r',y) \, d\lambda(r') \bigg)p(t,r,x)\Bigg\} 
		+ \frac {\vert \sigma(r)\vert^2} 2 \Delta_x\left[  p(t,r,x)\right] .
	\end{multline}

	The mean-field equations~\eqref{eq:MFESpace} are similar to those found in the setting of~\cite{touboulNeuralfields:11} but present an additional term related to the presence of a micro-circuit, showing the local averaging effects arising in our setting. Interestingly, this shows a kind of universal behavior across all possible choices of parameters $v(N)$ and $\beta(N)$, i.e. across possible local statistics of the topology. The limit equations are very complex: similarly to what was discussed in~\cite{touboulNeuralfields:11}, they resemble McKean-Vlasov equations but involve delays, spatially chaotic Brownian motions and an `integral over spatial locations' (in a sense that will be made clearer in the sequel). This is hence a very unusual stochastic equation we need to thoroughly study in order to ensure that these make sense and are well-posed. The existence and uniqueness of solutions to this mean-field equation are addressed in section~\ref{sec:ExistenceUniquenessSpace}, and the proof of the propagation of chaos and convergence of the network equations towards the solution of the mean-field equation is performed in section~\ref{sec:PropaChaSpace}.

	\section{Analysis of the mean-field equation}\label{sec:ExistenceUniquenessSpace}

	The mean-field equation~\eqref{eq:MFESpace} involves two unusual terms: a stochastic integral involving spatially chaotic Brownian motions and an integrated McKean-Vlasov mean-field term with delays. The spatially chaotic Brownian motion was introduced in~\cite{touboulNeuralfields:11}. In order to handle these equations, we start by introducing and discussing the functional spaces in which we are working, and the notion of solutions to these singular equations. 

 First of all, in order to make sense of the mean-field equation~\eqref{eq:MFESpace}, we need to show that the Lebesgue's integral over $\Gamma$ term is well defined. This integral involves the expectation of the process, so even though the solution is not measurable with respect to space, its expectation may be depending on the regularity of its law with respect to space. In this view, we consider the set $\mathcal{Z}_T$ of spatially chaotic processes $(X_t(r))$ defined for times $t\in[-\tau,T]$ that have the following continuity property:  there exists a \emph{coupled process} $\hat{X}_t(r)$ indexed by $r\in\Gamma$, such that for any fixed $r\in\Gamma$, $\hat{X}_t(r)$ has the same law as $X_t(r)$, and moreover, there exists a constant $C>0$  such that for any $(r,r')\in\Gamma^2$:
		\begin{equation}\label{eq:RegSpace}
			\Ex{\sup_{t\in [-\tau,T]}\left\vert \hat{X}_t(r)-\hat{X}_t(r')\right\vert}\leq C (\vert r-r'\vert + \sqrt{\vert r-r'\vert}).
		\end{equation}
		Note that in this case, for any Lipschitz-continuous function $\varphi$, the map $r\mapsto \Ex{\varphi(X_t(r))}$ is continuous (H\"older $1/2$). In particular, it is measurable. One can then define the Lebesgue's integral of it over $\Gamma$. Such processes are called
		 \emph{chaotic processes with regular law}. 
		\begin{remark}
			We note that in the absence of space-dependent delays, the process is more regular (Lipschitz-continuous). 
		\end{remark}

Moreover, stochastic processes $(X_t(r)) \in \mathcal{Z}_T$ are said square integrable for all $r\in\Gamma$ if:
		\[\sup_{r\in\Gamma} \Ex{\vert X_t(r)\vert^2}<\infty.\]

We further consider the subset $\Z_T^{\star}\subset \Z_T$ composed of processes $(X_t(r))$ that satisfy the following regularity in time: there exists $C>0$ such that:
		\[
			\Ex{\left \vert X_t(r)-X_{t'}(r)\right \vert }\leq C (\sqrt{\vert t-t'\vert} + \vert t-t'\vert).
		\]
Eventually, for a process $(X_t(r))_{r\in\Gamma}\in \mathcal{Z}_T$, we define the squared norm: 
				\begin{equation}\label{eq:Norm}
					\norm{T}{Z}=\int_{\Gamma}\Ex{\sup_{s\in[-\tau,T]}\vert{Z_s({r})}\vert^2}\,d\lambda(r)
				\end{equation} 
				and the $\mathbbm{L}^1$ norm:
				\begin{equation*}
					\Vert{Z}\Vert^1_T=\int_{\Gamma}\Ex{\sup_{s\in[-\tau,T]}\vert{Z_s({r})}\vert}\,d\lambda(r).
				\end{equation*}
				These clearly define norms on random variables indexed by $r\in\Gamma$, when identifying processes that are $\lambda\otimes\P$-a.s. equal. We denote $\mathcal{Z}_T^2$ the set of of random variables in $\mathcal{Z}_T$ such that $\norm{T}{Z}<\infty$. Note that this norm depends on $\lambda$ the distribution over $\Gamma$ of neurons. It is of course possible to define a norm using Lebesgue's measure on $\Gamma$, which would be in that case independent of the choice of $\lambda$. The two obtained measures will be of course equivalent since here we assumed that $\lambda$ was equivalent to Lebesgue's measure.

\begin{example}
	Let us start by giving a simple yet informative example of such process. Let $W_t(r)$ be a spatially chaotic Brownian motion, and consider $(\Delta_t(r))_{t\in[0,T], r\in\Gamma}$ a $\F_t$-progressively measurable real-valued process indexed by $r\in \Gamma$ that belongs to $\Z_T$ and which is independent of the collection of Brownian motions $(W_t(r))$. We denote by $\hat{\Delta}_t(r)$ the coupled process corresponding to the regularity condition. We assume that for any $(r,r')\in \Gamma$ we have 
			\begin{equation}\label{eq:BoundedDelta}
				\begin{cases}
					\Ex{\vert\Delta_s(r)\vert^2}< C<\infty, \text{ and}\\
					 \Ex{\vert \hat{\Delta}_t(r)-\hat{\Delta}_t(r')\vert^2}\leq C^2 (\vert r-r'\vert + \sqrt{\vert r-r'\vert})^2
				\end{cases} 
			\end{equation}
			Since for  any fixed $r\in\Gamma$, the process $W_t(r)$ is a standard Brownian motion, the process $N_t(r)$ defined by the stochastic integral:
			\[N_t(r):=\int_0^t \Delta_s(r)dW_s(r)\] 
			is well defined. It is spatially chaotic since for $r\neq r'$ the Brownian motions $W_t(r)$ and $W_t(r')$ and the processes $\Delta_t(r)$ and $\Delta_t(r')$ are independent. Moreover, they have a regular law in the sense of our definition. Indeed, let $(W_t)$ be a standard Brownian motion independent of $\hat{\Delta}_t(r)$. The process 
			\[\hat{N}_t(r):=\int_0^t \hat{\Delta}_s(r)dW_s\] 
			has the same law as $N_t(r)$, and moreover, 
			\[\Ex{\vert \hat{N}_t(r)-\hat{N}_t(r')\vert}\leq \left(\int_0^t \Ex{\vert  \hat{\Delta}_s(r)-\hat{\Delta}_s(r')\vert^2}\,ds\right)^{1/2}\leq \sqrt{T} C (\vert r-r'\vert+\sqrt{\vert r-r'\vert}).\] 
			The process $N_t(r)$ therefore belongs to $\mathcal{Z}_T$. Moreover, it is a square integrable martingale with quadratic variation $\int_0^t \Ex{\vert \Delta_s(r)\vert^2}\,ds$. This implies that for any $t<t'$ in $[0,T]$:
			\[\Ex{\vert N_t(r)-N_{t'}(r)} = \Ex{\vert\int_t^{t'} \Delta_s(r) dW_s(r)\vert}\leq \left(\int_{t}^{t'} \Ex{\vert\Delta_s(r)\vert^2}\,ds\right)^{1/2}\leq C \sqrt{t'-t}.\]
			
			The process therefore belongs to $\Z_T^{\star}$. 
			
			Note that this example illustrates an important fact. The process $N_t(r)$ involves two processes, $\Delta_t(r)$ and $W_t(r)$, and in order to build up the coupled $\hat{N}_t(r)$, we used the fact that we were able to find two processes $\hat{\Delta}_t(r)$ and $W_t$ such that the pairs $(\Delta_t(r),W_t(r))$ and $(\hat{\Delta}_t(r),W_t)$ had the same law (here, the two components are independent). This fact will be also prominent in the definition of the solutions to the mean-field equation.
\end{example}

\renewcommand{\theenumi}{(\roman{enumi})}

\begin{definition}\label{def:solution}
	A \emph{strong solution} to the mean-field equation~\eqref{eq:MFESpace} on the probability space $(\Omega,\F,\P)$, with respect to the chaotic Brownian motion $(W_t(r))$ and with an initial condition $\zeta^0_t\in \Z_0^{\star}$ is a spatially chaotic process $X=\{X_t(r); -\tau\leq 0 \leq T, r\in\Gamma\}\in \Z_T^{\star}$, i.e. with continuous sample paths and regular law, such that:
	\begin{enumerate}
		\item there exists a coupling $(W_t,\hat{\zeta}^0_t(r),\hat{X}_t(r) ; t\in [-\tau, T], r\in\Gamma)$, such that for any fixed $r\in \Gamma$ $(W_t,\hat{\zeta}^0_t(r),\hat{X}_t(r))\eqlaw(W_t(r),{\zeta}^0_t(r),{X}_t(r))$ and moreover:
		\[\begin{cases}
			\Ex{\sup_{t\in [-\tau,0]}\left\vert \hat{\zeta}^0_t(r)-\hat{\zeta}^0_t(r')\right\vert}\leq C (\vert r-r'\vert + \sqrt{\vert r-r'\vert}),\\
			\Ex{\sup_{t\in [-\tau,0]}\left\vert \hat{X}_t(r)-\hat{X}_t(r')\right\vert}\leq C (\vert r-r'\vert + \sqrt{\vert r-r'\vert}).
		\end{cases}\]
		This regularity ensures that for any Lipschitz-continuous map $\varphi:E\mapsto E$, the map $r\mapsto\Ex{\varphi(X_t(r))}$ is continuous. In particular, it is measurable and hence the integral $\int_{\Gamma}\Ex{\varphi(X_t(r))}d\lambda(r)$ can be computed in the usual (Lebesgue's) sense. 
		\item for any $r\in \Gamma$, $(X_t(r), t\in[-\tau,T])$ is a strong solution, in the usual sense (see~\cite[defintion 5.2.1]{karatzas-shreve:87}), i.e. it is adapted to the filtration $(\F_t)$, almost surely equal to $\zeta^0_t$ for $t\in [-\tau,0]$, and the equality:
		\[X_t(r)=\begin{cases}
				\zeta^0_t(r) + \int_0^t  f(r,s,X_s(r)) ds+ \int_0^t \sigma(r)   dW_s(r) + \int_0^t\bar{J}\Exp_{Z}[ b(X_s(r), Z_{s-\tau_s}(r)] \,ds\\
				\quad+ \int_0^t \int_{\Gamma} J(r,r')\Exp_{Z} [ b(X_s(r), Z_{s-\tau(r,r')} (r') )]d\lambda(r')\,ds ,\qquad  t>0\\
				 \zeta^0_t(r) ,\qquad t\in [-\tau, 0]\\
				 (Z_t)\eqlaw (X_t) \text{ independent of $(X_t)$ and $(W_t(\cdot))$ }
				\end{cases}\]
				holds almost surely.
	\end{enumerate} 
	
\end{definition}

	\begin{theorem}\label{thm:ExistenceUniquenessSpace}
		Let $(\zeta^0_t(r),\; t\in [-\tau,0],\; r\in\Gamma) \in \mathcal{Z}_{0}^2$ a square-integrable process with a regular law. The mean-field equation \eqref{eq:MFESpace} with initial condition $\zeta^0$ has a unique strong solution on $[-\tau,T]$ for any $T>0$. The solution belongs to $\Z_T^{\star}\cap \Z_T^2$.
	\end{theorem}
	\begin{proof}
		This theorem is proved through a usual fixed point argument for a map $\Phi$ acting on stochastic processes $X$ in $\mathcal{Z}_T^2$ defined by:
	\begin{align*}
		\Phi(X)_t(r) &= \begin{cases}
			\zeta^0_0(r) + \int_0^t  f(r,s,X_s(r)) ds+ \int_0^t \sigma(r)   dW_s(r) + \int_0^t\bar{J}\Exp_{Z}[ b(X_s(r), Z_{s-\tau_s}(r)] \,ds\\
					\quad+ \int_0^t \int_{\Gamma} J(r,r')\Exp_{Z} [ b(X_s(r), Z_{s-\tau(r,r')} (r') )]d\lambda(r')\,ds ,\qquad  t>0\\
					 \zeta^0_t(r) \qquad ,\qquad t\in [-\tau, 0]\\
					 (Z_t)\eqlaw (X_t) \text{ independent of $(X_t)$ and $(W_t(\cdot))$ }
		\end{cases}
	\end{align*}	
	We aim at building a sequence of processes by iterating the map $\Phi$ starting from a given initial process, and showing that this constitutes a Cauchy sequence, converging to the unique fixed point of the map, i.e. the unique solution of the mean-field equations. This classical scheme appears relatively complex to handle in our present case. Indeed, the construction of the sequence is not trivial, as we need to be able to integrate the expectation of a function of the processes, hence we need this expectation to be measurable with respect to space. Second is the fact that we aim at showing existence and uniqueness in a relatively strong sense (condition (ii) of the definition) valid for any $r\in\Gamma$ (and not $\lambda$-almost surely as would be the case under the norm~\eqref{eq:Norm}). 
	
	Let us start by showing that we can iterate the map $\Phi$. To this end, we analyze the processes $(Y_t(r))$, image of processes $(X_t(r))\in\Z_T^{\star} \cap \Z_T^2$ under the map $\Phi$. It is easy to see that $Y_t(r)$ is spatially chaotic. Let us start by showing that $Y_t(r)$ is square integrable (in what follows, $C$ denotes a constant independent of time, that may vary from line to line). We note that:
		\[\psi: (r,t,x)\mapsto f(r,s,x) + \bar{J}\Exp_{Z}[b(x,Z_{s-\tau_s}(r))] +\int_{\Gamma} J(r,r') \Exp_Z[b(x,Z_{s-\tau(r,r')})]d\lambda(r')\]
		is Lipschitz-continuous with respect to $(t,x)$ (with constant $K_f + L (\vert \bar{J}\vert+\Vert J \Vert_{\infty})$), and hence $\psi^2(r,t,x)\leq C(1+\vert x\vert^2)$. Standard inequalities allow showing that the value $N_t^Y(r)=\Ex{\sup_{s\in[-\tau,t]} \vert Y_s(r)\vert^2}$ satisfies the relationship:
			\[ N_t^Y(r)\leq 4 \Bigg(\Ex{\sup_{s\in[-\tau,0]} \vert \zeta^0_s(r)\vert^2} + T \,C\int_0^T (1+N_s^X(r))\,ds +4\,T \vert\sigma(r)\vert^2 \Bigg )\]
		which is finite under the assumption that $X_t(r)$ and $\zeta^0$ are square integrable. Note that this property readily implies, by application of Gronwall's lemma, that any possible solution is square integrable. 
		
		The regularity in time is then a direct consequence of this inequality and of the fact that the Lipschitz continuity of $\psi$ implies that $\vert \psi(r,t,x)\vert \leq C (1+t+\vert x \vert)\leq C(1+\vert x \vert)$. Indeed, for $t<t'$, we have
		\begin{align*}
			\Ex{\vert Y_t(r)-Y_{t'}(r)\vert } &\leq \int_{t}^{t'} \Ex{\vert \psi(r,s,X_s(r))\vert} ds + \Ex{\vert \sigma(r) (W_{t'}(r)-W_t(r))\vert}\\
			&\leq  C (1+N_T^X(r)^{1/2}) (t'-t) + \vert \sigma(r)\vert \sqrt{t'-t}.
		\end{align*}
		It therefore remains to show that $Y_t(r)$ is regular in law. Let $W_t$ be a standard Brownian motion, and assume that $(W_t,\hat{\zeta}_0^0(r),\hat{X}_t(r))$ is a coupling of $(W_t(r),\zeta_0^0(r), X_t(r))$ in the sense that they are equal in law for any fixed $r$, and that both $\hat{X}_t(r)$ and $\hat{\zeta}_0^0(r)$ have the regularity property~\eqref{eq:RegSpace} ($X$ is a process satisfying the assumptions of definition~\ref{def:solution}). We define $\hat{Y}_t(r)$ as:
		\begin{multline*}
			\hat{Y}_t(r)=\hat{\zeta}^0_0(r) + \int_0^t  f(r,s,\hat{X}_s(r)) ds+ \sigma(r)   W_t + \int_0^t\bar{J}\Exp_{Z}[ b(\hat{X}_s(r), Z_{s-\tau_s}(r)] \,ds\\
			+ \int_0^t \int_{\Gamma} J(r,u)\Exp_{Z} [ b(\hat{X}_s(r), Z_{s-\tau(r,u)} (u) )]d\lambda(u)\,ds
		\end{multline*}
		It is clear that this process has the same law as $Y_t(r)$ since  $(W_t, \hat{\zeta}_0^0(r),\hat{X}_t(r))\eqlaw (W_t(r), {\zeta}_0^0(r),{X}_t(r))$, and this obviously also holds for the processes $(W_t, \hat{\zeta}_0^0(r), \hat{Y}_t(r))$ and $(W_t(r),{\zeta}_0^0(r), {Y}_t(r))$. Let us denote $D_t^X=\Ex{\sup_{s\in [-\tau,t]} \vert \hat{X}_s(r)-\hat{X}_s(r')\vert}$. We have:
			\begin{align*}
						D_t^Y &\leq C(\vert r-r'\vert + \sqrt{\vert r-r'\vert})+\mathbbm{E}\Bigg[\sup_{s\in [-\tau,t]} \int_0^s \Bigg\{K_f(\vert r-r'\vert +\vert \hat{X}_u(r)-\hat{X}_u(r')\vert) + \vert \bar{J}\vert L  \Big(\vert \hat{X}_u(r)-\hat{X}_u(r')\vert  \\
						& + D_{u-\tau_s} \Big)+ K_{\Gamma} (1+\Vert b\Vert_{\infty})\vert r-r'\vert + \Vert J\Vert_{\infty} L \vert \hat{X}_u(r)-\hat{X}_u(r')\vert + \int_{\Gamma} \mathbbm{E}[\vert   \hat{X}_{u-\tau_s(r,v)}(v)-\hat{X}_{u-\tau_s(r',v)}(v)]d\lambda(v)\Bigg\}\,du\Bigg]\\
						&\leq \Big(C+T (K_{\Gamma}(1+\Vert b\Vert_{\infty})+K_f)\Big) \vert r-r'\vert +  C\Big(1+T\sqrt{K_{\Gamma}}\Vert J\Vert_{\infty} L \Big)\sqrt{\vert r-r'\vert} + \int_0^t \Big(K_f + 2\vert\bar{J}\vert L +\Vert J\Vert_{\infty}\Big) D_s^X\,ds
					\end{align*}
		and we conclude, using the assumption that $D^X_t \leq C (\vert r-r'\vert+\sqrt{\vert r-r'\vert})$, on the regularity of the law of the process $\hat{Y}_t(r)$. In particular, let us emphasize the fact that for  any $\varphi:E\mapsto \R$ a $1$-Lipschitz-continuous function,
\[\sup_{t\in[-\tau,T]} \vert \Ex{\varphi(X_t(r))-\varphi(X_t(r'))} \vert\leq C \Big(\vert r-r'\vert + \sqrt{\vert r-r'\vert}\Big),\] 
which implies that the expectation $\Exp_Z[b(X_t(r),Z_{t-\tau(r,r')}(r'))]$ is measurable with respect to the Borel algebra $\B(\Gamma)$ in $r'$, allowing to make sense of the integral over the space variable $r'$. Let us eventually remark that, again, by Gronwall's lemma, any possible solution has a coupled process satisfying the regularity condition~\eqref{eq:RegSpace}.

These properties ensure that we can make sense of the spatial integral term in the definition of $\Phi$ for iterates of that function. A sequence of processes can therefore be defined by iterating the map. 

We fix $X$ a process in ${\mathcal{Z}_T}$ satisfying the coupling assumptions above (related to definition~\ref{def:solution}), and build the sequence $X^k$ by induction through the recursion relationship $X^{k+1}=\Phi(X^k)$. We show that these processes constitute a Cauchy sequence for  $\Vert \cdot \Vert_{\mathcal{Z}_T}$ . This will not be enough for our purposes: we are interested in proving existence and uniqueness of solutions for all $r$. Equipped with the estimates on the distance~\eqref{eq:Norm}, we will come back to the sequence of processes at single locations, show that these also constitute a Cauchy sequence in the space of stochastic processes in $E$ (which is complete) and conclude. 
	
	Again, one needs to be careful in the definition of the above recursion and build recursively a sequence of processes  $(Z^k)$ independent of the collection of processes $(X^k)$ and having the same law as follows:
	\begin{itemize}
		\item $Z^0$ is independent of $X^0$ and has the same law as $X^0$
		\item for $k\geq 1$, $Z^{k}$ is independent of the sequence of processes $(X^{0},\cdots,X^k)$ and is such that the collection of processes $(Z^{0},\cdots,Z^{k})$ has the same joint law as $(X^{0},\cdots,X^k)$, i.e. $Z^k$ is chosen such as its conditional law given $(Z^0,\cdots,Z^{k-1})$ is the same as that of $X^k$ given $(X^0,\cdots,X^{k-1})$.
	\end{itemize} 
	
	Once all these ingredients have been introduced, it is easy to show that $M^k_t=\Vert{X^{k+1}-X^k}\Vert_t^1$ satisfies a recursion relationship, by decomposing this difference into the sum of elementary terms:
	\begin{align*}
		X^{k+1}_t(r)-X^{k}_t(r) &= \int_0^t \Big\{(f(r,s,X_s^{k}(r))-f(r,s,X_s^{k-1}(r)))\Big\}\,ds \\
		& \quad + \bar{J}\int_0^t \Exp_{Z} [ b(X_s^{k}(r), Z^{k}_{s-\tau_s} (r) )- b(X_s^{k-1}(r), Z^{k-1}_{s-\tau_s} (r)]\\
		 & \quad + \int_0^t \int_{\Gamma} J(r,r')\Big\{\Big( \Exp_{Z} [ b(X_s^{k}(r), Z^{k}_{s-\tau(r,r')} (r') )]\\
		&\qquad \qquad \qquad-\Exp_{Z} [ b(X^{k-1}_s(r), Z^{k-1}_{s-\tau(r,r')} (r') )] \Big)\Big\}d\lambda(r')\,ds \\
			 & =: A_t(r) + B_t(r)+ C_t(r)
	\end{align*}
	and checking that the following inequalities apply:
	\[\Vert{A}\Vert_t^1 \leq K_f \int_0^t \Vert{X^k-X^{k-1}}\Vert_s^1\,ds\]
	through the use of Cauchy-Schwarz inequality, 
	\[\Vert{B}\Vert_t^1\leq 2\vert \bar{J}\vert L \int_0^t \Vert{X^k-X^{k-1}}\Vert_s^1\,ds\]
	by standard McKean-Vlasov arguments, and
		\begin{align*}
				\Vert{C}\Vert_t^1 & = \Exp\Bigg[\int_{\Gamma}\sup_{s\in[0,t]}  \bigg \vert \int_{\Gamma}J(r,r')\int_0^s \Big(\Exp_{Z} [ b(X_u^{k}(r), Z_{u-\tau(r,r')}^{k} (r') ) - b(X_u^{k-1}(r), Z_{u-\tau(r,r')}^{k-1} (r') )]\Big) {du\,}d\lambda(r') \bigg\vert d\lambda(r)\Bigg]\\
			 & \leq \Vert J\Vert_{\infty} \;\int_{\Gamma^2} \int_0^t\Exp\Bigg[ \Exp_{Z} [\Big \vert b(X_u^{k}(r), Z_{u-\tau(r,r')}^{k} (r') ) - b(X_u^{k-1}(r), Z_{u-\tau(r,r')}^{k-1} (r') )]\Big\vert\Big) du\Bigg]d\lambda(r)d\lambda(r') \quad (CS)\\
				&\leq 2 \, L\Vert J\Vert_{\infty} \int_{\Gamma^2} \int_0^t \Ex{\vert X_s^{k}(r)-X_s^{k-1}(r)\vert}\,ds d\lambda(r)d\lambda(r')\quad \ref{Assump:LocLipschbSpace}\\
				&\leq 2 \;L\Vert J\Vert_{\infty} \int_0^t \Vert{X^{k}-X^{k-1}}\Vert_s^1\,ds
			\end{align*}
	These inequalities imply:
		\begin{align}
			\nonumber M_t^k &\leq K' \int_{0}^t M_s^{k-1}\,ds \\
			\label{eq:BoundaryMRec}&\leq \frac{(K' T)^k}{k!} M_T^0
		\end{align}
	with $K'=K_f+2L(\vert \bar{J}\vert + \Vert J\Vert_{\infty})$.
	
	Let us now  denote for $Z\in \mathcal{Z}_T$ the norm $\Vert Z(r) \Vert_t = \Ex{\sup_{s\in[0,t]} \vert Z_s(r) \vert}$ consider $N_t^k(r)=\Vert X^{k+1}(r)-X^{k}(r)\Vert_t$. Similar developments yield to the inequality:
	\begin{align*}
		N_t^k(r) & \leq \Vert A(r)\Vert_t + \Vert B(r)\Vert_t +\Vert C(r)\Vert_t\\
		& \leq (K_f +2 \vert \bar{J}\vert L + \Vert J\Vert_{\infty}L)  \int_0^t N_s^{k-1}(r)\,ds + L\Vert J\Vert_{\infty} \int_0^t M_s^{k-1}\,ds\\
		&\leq K_1 \int_0^t N_s^{k-1}(r)\,ds + K_2\int_0^t M_s^{k-1}\,ds
	\end{align*}
	where $K_1$ and $K_2$ correspond to the constants of the penultimate equation. We denote $K_3=\max(K_1,K_2,K')$ and $D=\max(N_T^0(r),M_T^0)$. By recursion and using equation~\eqref{eq:BoundaryMRec}, we obtain:
	\[N^k_T(r)\leq \frac{(K_3\,T)^k}{(k-1)!} D.\] 
	This implies that the processes $(X^k_t(r))$ for fixed $r$ constitute a Cauchy sequence in the space of stochastic processes. From this relationship, routine methods allow proving existence and uniqueness of fixed point for $\Phi$ (see e.g.~\cite[pp. 376--377]{revuz-yor:99}), and that this fixed point is adapted and almost surely continuous. Proving uniqueness of the solution using equation~\eqref{eq:BoundaryMRec} is then classical. 
	\end{proof}

We therefore proved that there exists a unique solution to the mean-field equation, which moreover is regular in space in the sense defined above. Of course, as stated, the solutions are discontinuous at all points $r\in\Gamma$. The proposition nevertheless ensures a form of regularity in law, which will be central in the sequel to prove averaging effects in the microcircuit.

	\section{The mean-field limit and propagation of chaos}\label{sec:PropaChaSpace}
	Now that we have introduced suitable spaces in which the mean-field equations are well-defined, and proved that the equation was well-posed, we are in a position to demonstrate the main result of the manuscript, namely the convergence in law of the solutions of the network equations~\eqref{eq:NetworkSpace} towards the equations~\eqref{eq:MFESpace}, and the fact that the propagation of chaos property occurs. We consider that the network equations have chaotic initial conditions with law continuous in space $(\zeta^0_t(r))\in \Z_0^{\star} \cap \Z_0^2$. In detail, the initial conditions of the $N$ neurons in the network are considered independent processes $(\zeta^{i}_t)\in \M^2([-\tau,0], E)$ (the space of square integrable processes from $[-\tau,0]$ on $E$) with law equal to  $(\zeta^{0}_t(r_{i}))$.

	Our convergence result raises several difficulties compared to more standard models:
	\begin{itemize}
		\item First is the fact that at the micro-circuit scale, there will be a local averaging principle (yielding the convergence towards a local term $\bar{J}\Exp_Z[b(X_t(r),Z_{t-\tau_s}(r))]$. This property is not classical: indeed, in the network equation, the microcircuit interaction term is $\sum_{j\in \V(i)} b(X^i_t,X^j_{t-\tau_{ij}})$, and therefore involve the state of neurons located at different places on $\Gamma$ and different delays. The convergence will be handled using (i) the fact that in the limit considered, the neurons belong to the microcircuit collapse at a single space location, and (ii) regularity properties of the law of the solution as a function of space and time. This convergence will be the subject of lemma~\ref{lem:ConvMicro}.
		\item Second, the macro-circuit interaction term involve delocalized terms across the neural field. The sum will be shown to converge to a non-local averaged term involving an integral over space. This will be proved through the use of lemma~\ref{lem:ConvMacro} and~\ref{lem:SumChi}.
	\end{itemize}
	Moreover, we will prove our convergence through a non-classical coupling method that we now describe.
	
	\subsection{Coupling between network equations and the continuous mean-field equation}
	Let us now fix a configuration $\A_N$ of the network. The neuron labeled $i$ in the network is driven by the $m$-dimensional Brownian motion $(\tilde{W}^{i}_t)$, and has the initial condition $\zeta^i\in\M(\Ct)$. We aim at defining a spatially chaotic Brownian motion $W^i_t(r)$ on $\R^{m}$ such that the standard Brownian motion $(W^i_t(r_i))$ is equal to $(\tilde{W}^i_t)$, and proceed as follows. Let $(W_t(r))_{t\in[0,T], r\in\Gamma}$ be a ${m}$-dimensional spatially chaotic Brownian motion independent of the processes $(\tilde{W}_t^{j})$. The process $W^i_t(r)$ defined by the coupling:
	\[
	\begin{cases}
		(W^i_t(r)) \;=\; (W_t(r))  \qquad r\neq r_{i}\\
		(W^i_t(r_{i})) \;=\; (\tilde{W}^i_t)
	\end{cases}
	\]
	is clearly a spatially chaotic Brownian motion, and will be used to construct a particular solution of the mean-field equations. In order to completely define a solution of the mean-field equations, we need to specify an initial condition, and aim at coupling it to the initial condition of neuron $i$. To this end, we define a spatially chaotic process $(\tilde{\zeta}^0_t(r)) \in \Z_0^{\star} \cap \Z_0^2$ equal in law to $(\zeta^0_t(r))$ and independent of $(\zeta^i_t)$, and define a coupled process $({\zeta}^{i,0}_t(r)) \in \mathcal{Z}_0$ as:
	\[\begin{cases}
		\zeta^{i,0}_t(r) = \tilde{\zeta}^0_t(r) \qquad r\neq r_{i}\\
		\zeta^{i,0}_t(r_{i}) = \zeta^i_t.
	\end{cases}\]

	\noindent Here again, it is clear that this process is spatially chaotic, i.e. that for any $r\neq r'$, the processes $\zeta^{i,0}_t(r)$ and $\zeta^{i,0}_t(r')$ are independent, and that $\zeta^{i,0}_t(r)$  has the law of $\zeta^0_t(r)$. 

	\medskip

	Now that these processes have been constructed, we are in a position to define the process $(\bar{X}^i_t)$ as the unique solution of the mean-field equation~\eqref{eq:MFESpace}, driven by the spatially chaotic Brownian motion $(W_t^i(r))$ and with the spatially chaotic initial condition $(\zeta^{i,0}_t(r))$: 
	\begin{equation*}
		\left \{
		\begin{array}{lll}
				d\bar{X}^i_t(r) &= \displaystyle{f(r,t,\bar{X}^i_t(r))\,dt + \int_{\Gamma} J(r,r')\Exp_Z[b(\bar{X}^i_t(r),Z_{t-\tau(r,r')}(r'))]d\lambda(r')\,dt} \\
				& \displaystyle{\quad +\bar{J}\Exp_{Z}[b(X_t(r),Z_{t-\tau_s}(r))]\,dt+ \sigma(r)\, dW^i_t(r) \qquad \text{for } t\geq 0}\\
				\\
			\bar{X}^i_t(r) & = \zeta^{i,0}_t (r) \qquad \text{for }  t\in [-\tau, 0]\\
			\\
			(Z_t)&\eqlaw (\bar{X}^i_t) \in \M \quad \text{ independent of $(\bar{X}^i_t)$ and $(W^{i}_t(\cdot))$}.
		\end{array}
		\right .
	\end{equation*}
	The same procedure applied for all $j\in\N$ allows building a collection of independent stochastic processes $(\bar{X}^j_t(r))_{j=1\ldots N} \in \mathcal{Z}_{T}$. These are clearly independent of the configurations of the finite-size network.
	Let us denote by $m(t,r)$ the probability distribution of $\bar{X}_t(r)$ solution of the mean-field equation \eqref{eq:MFESpace}. As previously, the process $(Z_t(r))$ generically denotes a process belonging to $\mathcal{Z}_T$ and distributed as $m$. 

\subsection{Local (micro-circuit) averaging}
	We start by analyzing the local averaging property on the micro-circuit. This is the subject of the following lemma. 
	\begin{lemma}\label{lem:ConvMicro}
		There exists a positive constant $K_1$ such that, for any $N$ sufficiently large, averaged across all configurations $\A_N$:
		\[\E\left\{\Ex{\left \vert\frac{1}{v(N)} \sum_{j\in \V(i)} b(\bar{X}^i_t(r_i),\bar{X}^j_{t-\tau(r_i,r_j)}(r_j)) - \Exp_{Z}[b(\bar{X}^i_t(r_i),{Z}_{t-\tau_s}(r_i))]\right \vert}\right\} \leq K_1 \sqrt{\left({\frac{v(N)}{N}}\right)^{\frac{1}{d}}+\frac{1}{{v(N)}}} \]
	\end{lemma}
	\begin{proof}
		Conditioned on $r_i$, the set $(\bar{X}^j_{t-\tau_{ij}}(r_j))$ is a collection of independent identically distributed random variables. The map $x\mapsto b(\bar{X}^i_t,x)$ is Lipschitz continuous. Therefore, the regularity properties proved in theorem~\ref{thm:ExistenceUniquenessSpace} ensure that we have (in what follows, $K_1$ denotes a constant, independent of $N$, that may change from line to line):
		\begin{equation}\label{eq:contlemma}
			\vert \Exp_{{Z}}[b(\bar{X}^i_t(r),{Z}_{t-\tau_{ij}}(r_j)]-\Exp_{{Z}}[b(\bar{X}^i_t(r),{Z}_{t-\tau_{s}}(r_i)] \vert \leq K_1(\sqrt{\vert \tau_{ij}-\tau_s \vert} + \vert r_j-r_i\vert+ \sqrt{\vert r_j-r_i\vert})=K_1(\sqrt{d_{ij}}+d_{ij}).
		\end{equation}
		For almost any configuration $\A_N$ and any $j\in\V(i)$, we have seen that for $N$ sufficiently large, by application of proposition~\ref{lem:SizeMicro}, the distances $d_{ij}$ are small (of order $(v(N)/N)^{1/d}$). 
		Moreover, we have:
		\begin{align*}
			&\frac{1}{v(N)} \sum_{j\in \V(i)} b(\bar{X}^i_t(r_i),\bar{X}^j_{t-\tau(r_i,r_j)}(r_j)) - \Exp_{Z}[b(\bar{X}^i_t(r_i),{Z}_{t-\tau_s}(r_i))] \\
			&\quad = \frac{1}{v(N)} \sum_{j\in \V(i)}\left( b(\bar{X}^i_t(r_i),\bar{X}^j_{t-\tau(r_i,r_j)}(r_j)) - \E[\Exp_{Z}[b(\bar{X}^i_t(r_i),Z_{t-\tau_{ij}}(r_j))]]\right)\\
			&\qquad + \frac{1}{v(N)} \sum_{j\in \V(i)}\E[\Exp_{Z}[b(\bar{X}^i_t(r_i),Z_{t-\tau_{ij}}(r_j))]] -  \Exp_{Z}[b(\bar{X}^i_t(r_i),{Z}_{t-\tau_s}(r_i))]
		\end{align*}
		For any measurable function $F$, the quantity $\E[\Exp_{Z}[F(Z_{t-\tau_{ij}(r_j)})]]$ is precisely the average of the random variable $F(\bar{X}^j_{t-\tau_{ij}}(r_j))$. Therefore, a quadratic control argument (see e.g.~\cite[Theorem 1.4.]{sznitman:89}) allows to show that the first term is of order $1/\sqrt{v(N)}$, in the sense that:
		\[\E[\Exp[\frac{1}{v(N)} \sum_{j\in \V(i)}\left( b(\bar{X}^i_t(r_i),\bar{X}^j_{t-\tau(r_i,r_j)}(r_j)) - \E[\Exp_{Z}[b(\bar{X}^i_t(r_i),Z_{t-\tau_{ij}(r_j)})]]\right)]]\leq \frac{K_1}{\sqrt{v(N)}}.\]
	This argument consists in showing that the expectation of the square of the sum is of order $1/v(N)$, which is performed by showing that (i) the terms of the sum are centered (i.e. that the expectation term introduced -- which was chosen to this purpose-- is precisely the expectation with respect to $Z$ of $b(\bar{X}^i_t,Z)$ for $Z$ equal in law to $\bar{X}^j_{t-\tau_{ij}}(r_j)$ which all have the same law) and (ii) using Cauchy-Schwarz inequality to bound the term by the square root of the expectation of the squared sum, developing the square and showing that the number of null terms is bounded by some constant multiplied by $v(N)$. This argument is not developed here as it will be the core of the proof of lemma~\ref{lem:ConvMacro}.
	
	The second term is handled by using the control given by equation~\eqref{eq:contlemma} and the result of proposition~\ref{lem:SizeMicro}, ensuring that
		\begin{multline*}
			\E\Big[\Exp\big[\frac{1}{v(N)} \sum_{j\in \V(i)}\E[\Exp_{Z}[b(\bar{X}^i_t(r_i),Z_{t-\tau_{ij}(r_j)})]] -  \Exp_{Z}[b(\bar{X}^i_t(r_i),{Z}_{t-\tau_s}(r_i))]\big]\Big]\\\leq K_1\Bigg(\sqrt{\left({\frac{v(N)}{N}}\right)^{\frac{1}{d}}+\frac 1 {v(N)}}
			+\left({\frac{v(N)}{N}}\right)^{\frac{1}{d}}+\frac 1 {v(N)}\Bigg).
		\end{multline*}
		Put together, the two last estimates yield the desired result. 
	\end{proof}

\subsection{Continuous (macro-circuit) averaging}
\begin{lemma}\label{lem:ConvMacro}
	The coupled macroscopic interaction term converges towards a non-local mean-field term with speed $\frac{1}{\sqrt{N\beta(N)}}$, in the sense that there exists a constant $K_2>0$ independent of $N$ such that:
	\[\E\Big[\Exp\big[\left \vert\frac 1 {N\beta(N)} \sum_{j=1}^N J(r_i,r_j)\chi_{ij}b(\bar{X}^i_t(r_i),\bar{X}^j_{t-\tau_{ij}(r_j)}) - \int_{\Gamma} J(r_i,r) \Exp_{Z}[b(\bar{X}^i_t(r_i),Z_{t-\tau(r_i,r)}(r))]d\lambda(r)\right \vert\big]\Big]\leq \frac{K_2}{\sqrt{N\beta(N)}}.\]
\end{lemma}
\begin{proof}
	Conditioned on the location $r_i$ of neuron $i$, the collection of $\Omega\times \Omega'$-random variables $(r_j,\chi_{ij},\bar{X}^j_{t-\tau_{ij}}(r_j))$ are independent and identically distributed. The sum
	\[\frac 1 {N\beta(N)} \sum_{j=1}^N J(r_i,r_j)\chi_{ij}b(\bar{X}^i_t(r_i),\bar{X}^j_{t-\tau_{ij}}(r_j))\]
	is therefore, conditionally on $\bar{X}^i_t$ and $r_i$, the sum of independent and identically distributed processes, with finite mean and variance (since $b$ is a bounded function). The expectation of each term in the sum, conditionally on $\bar{X^i_t}$ and $r_i$, is equal to:
	\[\beta(N)\int_{\Gamma} J(r_i,r) \Exp_Z[b(\bar{X}^i_t(r_i),Z_{t-\tau(r_i,r)}(r))] d\lambda(r).\]
	Let us denote by
	\[\Phi_{ij}(x,y)=\frac{\chi_{ij}}{\beta(N)} J(r_i,r_j)  b(x,y) - \int_{\Gamma} J(r_i,r') \int_E b(x,z)m(t-\tau(r_i,r'),r',dz)\,d\lambda(r').\]
	The term under consideration is simply the empirical average $\frac 1 N \Phi_{ij}(\bar{X}^i_t,\bar{X}^j_{t-\tau_{ij}})$, and conditionally on $\bar{X}^i_t$ and $r_i$, the terms are independent, identically distributed, centered $(\Omega\times \Omega')$-random variables with second moment:
	\[\textrm{Var}[\Phi_{ij}\vert \bar{X}^i_{t},r_i]\leq \frac{1}{\beta(N)}M_2\]
	where $M_2$ is a finite constant independent of $N$, $r_i$ and $\bar{X}^i_t$. 
	Let us denote by $\hat{\Exp}_i[\cdot]$ the expectation on $\Omega\times \Omega'$ conditioned on $r_i$ and $\bar{X^i_t}$. We have:
	\begin{align*}
		\hat{\Exp}_i\Big[\big \vert\frac 1 N \sum_{j=1}^N \Phi_{ij}\big \vert\Big] &\leq \sqrt{\hat{\Exp}_i\Big[\big(\frac 1 N \sum_{j=1}^N \Phi_{ij}\big)^2\Big]}\\
		&\leq \frac 1 N \sqrt{\sum_{j=1}^N \hat{\Exp}_i\Big[\Phi_{ij}^2\Big]} \leq \sqrt{\frac{M_2}{N\beta(N)}}.
	\end{align*}
	Thanks to the fact that $M_2$ is independent of $\bar{X^i_t}$, we conclude that:
	\[\E\Big[\Exp\big[ \vert \frac 1 N \sum \Phi_{ij} \vert\big] \Big] = \Exp\Big[\hat{\Exp}_i\big[\vert \frac 1 N \sum_{j=1}^N \Phi_{ij}\vert\big]\Big\vert r_i\Big]\leq \sqrt{\frac{M_2}{N\beta(N)}}\]
\end{proof}
\subsection{The multiscale convergence result}
Now that we have analyzed the local and macroscopic interaction terms defined with the coupled processes, we are in a position to demonstrate our main result, namely the full multiscale convergence result. 
	\begin{theorem}\label{thm:PropagationChaosSpace}
	Let $i\in \N$ a fixed neuron in the network. Under the assumptions~\ref{Assump:LocLipschSpace}-\ref{Assump:SpaceContinuity}, for almost all configuration $\A_N$ of the neuron locations $(r_{i}, i\in\N)$ and connectivity links $(\chi_{ij},(i,j)\in\N^2)$, the process $(X^{i,\A_N}_t, t\leq T)$ solution of the network equations \eqref{eq:NetworkSpace} converges in law towards the process $(\bar{X}_t(r_{i}), t\leq T)$ solution of the mean-field equations \eqref{eq:MFESpace} with initial condition $(\zeta^0_t(r))$ and moreover, the speed of convergence is given by:
		\begin{equation}\label{eq:PropchaosSpace}
			\max_{i=1\cdots N}\mathcal{E}\left( \Exp\Big [\sup_{-\tau\leq s\leq T} \vert X^{i,\A_N}_s - \bar{X}^i_s(r_{i})\vert \Big]\right) = O\left(\sqrt{\left(\frac{v(N)}{N}\right)^{\frac 1 {d}}+\frac 1 {{v(N)}}}+\frac 1 {\sqrt{N\beta(N)}}\right)
		\end{equation}
	\end{theorem}
	\begin{remark}
	The notation $\mathcal{E}$ denotes the expectation on $(\Omega',\F',\Pc)$ the distribution of network configurations $\A_N$, i.e. space locations $(r_i)_{i\in\N}$ and connectivity links $(\chi_{ij})_{i,j \in \N}$. The expectation $\E[\Exp[\cdot]]$ is therefore the global expectation, i.e. on $(\Omega\times\Omega',\F\otimes \F',\P\otimes \Pc)$. The result shows that the expectation tends to zero. This implies quenched convergence (i.e. for almost all configuration $\A_N$) along subsequences. In detail, the speed of converge $S(N)$ announced in the theorem (on the righthand side of equation~\eqref{eq:PropchaosSpace}) allows to define subsequences (i.e. a sequence of network size $N_n=\varphi(n)$) for which we have almost sure convergence. These sequences are such that Borel-Cantelli lemma can be applied, namely subsequences extracted through a strictly increasing application $\varphi:\N\mapsto \N$ such that $S(\varphi(N))$ is summable. 
	\end{remark}
	
	We prepare for the proof by demonstrating the following fine estimate that will be used to control configurations with more links than the expected value:
		\begin{lemma}\label{lem:SumChi}
			Under our assumptions, for any $i\in\{1,\cdots,N\}$ and $\gamma>1$, we have for $N$ sufficiently large:
			\[\E\left(\frac 1 {N\beta(N)} \sum_{j=1}^N \chi_{ij} \ind {\mathcal{D}_{\gamma}}\right) \leq \frac 1 {\beta(N)} \exp\left(-\frac 1 2 \big(\gamma\log(\gamma)+1-\gamma\big)N\beta(N)\right)\]
			where $\mathcal{D}_{\gamma}=\{\omega'\in\Omega' ; \sum_{j} \chi_{ij} > \gamma N \beta(N)\}$. 
			\end{lemma}
			\begin{proof}
				In order to demonstrate the result, we make use of Chernoff-Hoeffding theorem~\cite{hoeffding1963probability} controlling the deviations from the mean of Bernoulli random variables\footnote{The theorem states that the sum of iid Bernoulli random variables $X_j$ with fixed mean $p$, is such that, for any $\varepsilon>0$,
				\[\P\left[\sum_{j=1}^m X_j > p+\varepsilon\right]\leq \left(\left(\frac p {p+\varepsilon}\right)^{p+\varepsilon}\left(\frac {1-p} {1-p-\varepsilon}\right)^{1-p-\varepsilon}\right)^m.\] 
			}. Let us denote by $S_N=\sum_{j=1}^N\chi_{ij}$. This is a binomial variable of parameters $(N,\beta(N))$. Chernoff-Hoeffding theorem ensures that:
	\begin{align*}
		\P\Big[S_N\geq \gamma N \beta(N)\Big] &\leq \left(\left(\frac{\beta(N)}{\gamma\beta(N)}\right)^{\gamma \beta(N)}\left(\frac{1-\beta(N)}{1-\gamma\beta(N)}\right)^{1-\gamma \beta(N)}\right)^N\\
		&\leq \exp\left(-\gamma \log(\gamma) N \beta(N) + N\big(1-\gamma\beta(N)\big)\Big(\log\big(1-\beta(N)\big)-\log\big(1-\gamma\beta(N)\big)\Big)\right).
	\end{align*}
	Using a Taylor expansion of the logarithmic terms for large $N$ (using the fact that $\beta(N)$ tends to zero at infinity), it is easy to obtain 
	\[\P\Big[S_N\geq \gamma N \beta(N)\Big] \leq \exp\left( (-\gamma\log(\gamma)+\gamma-1) N\beta(N) + O(N\beta(N)^2)\right)\]
	Note that for $\gamma>1$, the quantity $(\gamma\log(\gamma)-\gamma+1)$ is strictly positive. We therefore have, for $N$ sufficiently large, that the probability is bounded by:
	\[\P\Big[S_N\geq \gamma N \beta(N)\Big] \leq \exp\left( -\frac 1 2 (\gamma\log(\gamma)-\gamma+1) N\beta(N))\right).\]
	 This allows to conclude the lemma as follows. It is clear by definition that $S_N\leq N$. Therefore,
	\begin{align*}
		\E\left(\frac 1 {N\beta(N)} \sum_{j=1}^N \chi_{ij} \ind {\mathcal{D}_{\gamma}}\right) & \leq \frac 1 {\beta(N)}\P\Big[S_N\geq \gamma N \beta(N)\Big]
	\end{align*} 
	yielding the desired result.
			\end{proof}
			We are now in a position to perform the proof of theorem~\ref{thm:PropagationChaosSpace}.
	\begin{proof}
	The proof is based on evaluating the distance $\Exp [\sup_{-\tau \leq s\leq T} \vert X^{i,N}_s - \bar{X}^i_s(r_i)\vert^2 ]$, and breaking it into a few elementary, easily controllable terms. A substantial difference with usual mean-field proofs is that network equations correspond to processes taking values in $E$ in which the interaction term is sum over a finite number of neurons in the network equation, while the mean-field equation is a spatially extended equation with an effective interaction term involving an integral over $\Gamma$. This will be handled using the result of lemma~\ref{lem:ConvMacro}.

We introduce in the distance coupled interaction terms that were controlled in lemmas~\ref{lem:ConvMicro} and~\ref{lem:ConvMacro} and obtain the following elementary decomposition (each line of the righthand side corresponds to one term of the decomposition, $A_t(N)-E_t(N)$):
			\begin{align*}
				\nonumber &X^{i,\A_N}_t-\bar{X}^i_t(r_{i}) = \int_0^t (f(r_{i},s,X^{i,\A_N}_s)-f(r_{i},s,\bar{X}^i_s(r_{i}))) \, ds \\
				 &\quad +  \; \frac{\bar{J}}{v(N)} \sum_{j\in \V(i)} \int_0^t \Big(b(X^{i,\A_N}_s,X^{j,\A_N}_{s-\tau(r_{i},r_{j})}) -b(\bar{X}^i_s(r_{i}),\bar{X}^j_{s-\tau(r_{i},r_{j})}(r_j))\Big) \,ds\\
				&\quad  + \int_0^t \Big(\frac{\bar{J}}{v(N)} \sum_{j\in\V(i)} b(\bar{X}^i_s(r_{i}),\bar{X}^j_{s-\tau(r_{i},r_{j})}(r_j)) - \Exp_Z[ b(\bar{X}^i_s(r_{i}),Z_{s-\tau_s}(r_{j}))]\Big)\, ds\\
				 &\quad + \frac{1}{N\beta(N)} \sum_{j=1}^{N} \int_0^t J(r_i,r_j)\chi_{ij} \Big(b(X^{i,\A_N}_s,X^{j,\A_N}_{s-\tau_{ij}})-b(\bar{X}^{i}_s(r_{i}),\bar{X}^{j}_{s-\tau_{ij}}(r_{j})) \Big)\,ds\\
				 &\quad + \int_0^t \Big(\frac{1}{N\beta(N)} \sum_{j=1}^{N} J(r_i,r_j)\chi_{ij} b(\bar{X}^{i}_s(r_{i}),\bar{X}^{j}_{s-\tau_{ij}}(r_{j})) -\int_{\Gamma} J(r_i,r') \Exp_Z[b(\bar{X}^i_s(r_{i}),Z_{s-\tau(r_{i},r')}(r'))]d\lambda(r')\Big)\,ds\\		
				 &\nonumber \qquad =: A^i_t(N)+B^i_t(N)+C^i_t(N)+D^i_t(N)+E^i_t(N)
			\end{align*}
			It is easy to show, using assumptions~\ref{Assump:LocLipschSpace} and~\ref{Assump:LocLipschbSpace}, that the terms $A_t^i(N)$ and $B_t^i(N)$ satisfy the inequalities:
			\begin{align*}
					\Exp\Big[\sup_{-\tau\leq s\leq t} \vert A^i_s(N) \vert\Big] & \leq K_f\,\int_0^{t} \Exp\Big[\sup_{-\tau\leq u\leq s} \vert X_u^{i,\A_N}-\bar{X}_u^i(r_{i})\vert \Big]\, ds\\
					\max_{i=1\cdots N}\mathcal{E}\Big[\Exp\Big[\sup_{-\tau\leq s\leq t} \vert B_s^i(N) \vert\Big]\Big] & \leq \frac{v(N)+1}{v(N)} L\, \int_0^{t}  \max_{k=1\cdots N}\E\Big[\Exp\Big[\sup_{-\tau\leq u\leq s}\vert X^{k,\A_N}_u-\bar{X}^k_u(r_{i}) \vert\Big]\Big] \, ds
			\end{align*}
			The term $D^i_t$ requires to be handled with care, because of the sparsity of the macrocircuit. Indeed, this term involves the sum of $N$ random variables and is rescaled by $1/N\beta(N)$. As we assumed $\beta(N)\to 0$ in order to account for the sparsity in the macrocircuit, most terms in the sum are equal to zero\footnote{Non-zero terms correspond exactly to the neurons $j$ such that $\chi_{ij}=1$}. We have:
			\[\Exp\Big[\sup_{-\tau\leq s \leq t} \vert D^i_s \vert \Big]\leq \Vert J\Vert_{\infty} \frac 1 {N\beta(N)} \sum_{j} \chi_{ij} \int_0^t \Exp\big[\sup_{0 \leq u \leq s} \vert b(X^{i,\A_N}_u, X^{j,\A_N}_{u-\tau_{ik}})-b(\bar{X}^i_u,\bar{X}^j_{u-\tau_{ik}}) \vert \big]\,ds\]
			This expression shows how critical the singular sparse coupling is to our estimates. Indeed, the random variable $\frac 1 {N\beta(N)} \sum_{j} \chi_{ij}$ almost surely tends to $1$ as $N$ goes to infinity, but it can reach very large values (up to $1/\beta(N)$ which diverges as $N$ goes to infinity). Configurations $\A_N$ for which the sum is large are increasingly improbable, but for these configurations, the deterministic scaling $1/(N\beta(N))$ is not fast enough to overcome the divergence of the input term. There is therefore a competition between the probability of having configurations with large values of $\sum_{j}\chi_{ij}$ and the divergence of the solutions. However, in the present case, this control will be possible using the estimate of the probability that the number of links exceeds $\gamma N \beta(N)$ using the result of lemma~\ref{lem:SumChi}. Indeed, fixing $\gamma>1$, and distinguishing whether $\sum_j\chi_{ij}\leq \gamma N \beta(N)$ or not, we obtain:
			\begin{align*}
				\E\Big[ \Exp\Big[\sup_{-\tau\leq s \leq t} \vert D^i_s \vert \Big]\Big] &\leq  2 \gamma \; L \; \Vert J \Vert_{\infty} \int_0^t \max_{k=1\cdots N}\E\Big[ \Exp\big[\sup_{-\tau \leq u \leq s} \vert X^{k,\A_N}_{u}-\bar{X}^k_{u} \vert \big]\Big]\,ds\\
			& \qquad + 2 \Vert b\Vert_{\infty}\Vert J\Vert_{\infty} \E\left(\frac 1 {N\beta(N)} \sum_{j} \chi_{ij} \ind {\mathcal{D}_{\gamma}}\right) 
			\end{align*}
			where $\mathcal{D}_{\gamma}=\{\omega'\in\Omega' ; \sum_{j} \chi_{ij} > \gamma N \beta(N)\}$ as defined in lemma~\ref{lem:SumChi}. By application of this lemma, and using the fact that the second term of the upper bound is negligible compared to $\frac 1 {\sqrt{N\beta(N)}}$, we conclude that:
			\begin{align*}
				\max_{k=1\cdots N}\E\Big[ \Exp\Big[\sup_{-\tau\leq s \leq t} \vert D^k_s \vert \Big]\Big] &\leq  2 \gamma \; L \; \Vert J \Vert_{\infty} \int_0^t \max_{k=1\cdots N}\E\Big[ \Exp\big[\sup_{-\tau \leq u \leq s} \vert X^{k,\A_N}_{u}-\bar{X}^k_{u} \vert \big]\,ds
+ \frac{K_c}{\sqrt{N\beta{(N)}}}. 
			\end{align*}
			
			We are left with controlling the terms $C^i_t$ and $E^i_t$. These consist of sums only involving the coupled processes, and were analyzed in the previous sections. By direct application of the results of lemmas~\ref{lem:ConvMicro} and~\ref{lem:ConvMacro}, we have:
			\[\begin{cases}
				\displaystyle{\max_{k=1\cdots N}\E\Big[ \Exp\Big[\sup_{-\tau\leq s \leq t} \vert C^k_t \vert \Big]\Big]} &\leq \displaystyle{K_1\sqrt{\left(\frac{v(N)}{N}\right)^{\frac{1}{d}}+\frac 1 {{v(N)}}}}\\
				\displaystyle{\max_{k=1\cdots N}\E\Big[ \Exp\Big[\sup_{-\tau\leq s \leq t} \vert E^k_t \vert \Big]\Big]} &\leq \displaystyle{\frac{K_2}{\sqrt{N\beta{(N)}}}}
			\end{cases}\]

			All together, we hence have, for some constants $\tilde{K}>0$ and $\tilde{K}'>0$ independent of $N$ 
			\[M_t=\max_{i=1\cdots N}\mathcal{E}\left( \Exp\Big [\sup_{-\tau\leq s\leq t} \vert X^{i,N}_s - \bar{X}^i_s(r_{i})\vert \Big]\right),\]
			the inequality:
			\[M_t\leq \tilde{K}\int_0^t M_s \,ds + \tilde{K}'\left(\sqrt{\left(\frac{v(N)}{N}\right)^{\frac{1}{d}}+\frac 1 {{v(N)}}} + \frac{1}{\sqrt{N\beta{(N)}}}\right)\]
			which proves the theorem by application of Gronwall's lemma. 
	\end{proof}

	\begin{corollary}\label{cor:PropaChaosSpace}
		Let $l\in \N^*$ and fix $l$ neurons $(i_1,\cdots,i_l) \in \N^*$. Under the assumptions of theorem~\ref{thm:PropagationChaosSpace}, the process $(X^{i_1,\A_N}_t, \cdots, X^{i_l,\A_N}_t, -\tau\leq t \leq T)$ converges in law towards $m_t{(r_{p(i_1)})}\otimes \cdots \otimes m_t{(r_{p(i_l)})}$. 
	\end{corollary}

	\begin{proof}
		We have:
		\begin{align*}
			&\mathcal{E}\left(\Exp \left[ \sup_{-\tau\leq t \leq T} \left\vert (X^{i_1,\A_N}_t, \cdots, X^{i_l,\A_N}_t) - (\bar{X}^{i_1}_t, \cdots, \bar{X}^{i_l}_t)\right\vert^2 \right]\right) \\
			&\quad  \leq l \max_{k=1\cdots N}\mathcal{E}\left(\Exp \left[ \sup_{-\tau\leq t \leq T} \left\vert X^{k,\A_N}_t-\bar{X}^{k}_t\right\vert^2 \right]\right)\\
		\end{align*}
		which tends to zero as $N$ goes to infinity, hence the law of $(X^{i_1,\A_N}_t, \cdots, X^{i_l,\A_N}_t, -\tau\leq t \leq T)$  converges towards that of 
		$(\bar{X}^{i_1}_t, \cdots, \bar{X}^{i_l}_t, -\tau\leq t \leq T)$ which is equal by
		 definition to $m(t,r_{p(i_1)})\otimes \cdots \otimes m(t,r_{p(i_l)})$.
	\end{proof}
\section{Discussion}
The dynamics of neuronal networks in the brain lead us to analyze a class of spatially extended networks which display multiscale connectivity patterns that are singular in at least two aspects:
\begin{itemize}
	\item the network display local dense connectivity patterns in which neurons are connected to their $v(N)$-nearest neighbors, where $v(N)=o(N)$. 
	\item the macro-circuit was also singular, in the sense that the probability of two neurons $i$ and $j$ to be connect tends towards zero. This is very far from usual mean-field models that consider full connectivity patterns, or partial connectivity patterns proportional to the network size~\cite{stannat-lucon:13,touboul-hermann-faugeras:11,graham2009interacting}. In these cases, the convergence is substantially slower, and the rescaling actually required thorough controls on the number of incoming connections to each neurons.  
\end{itemize}
{The introduction of local microcircuits with negligible size was suggested in~\cite{stannat-lucon:13}, in the context of the fluctuations induced by the microcircuit. Our scaling, motivated by characterizing the macroscopic activity at the scale of the neural field $\Gamma$, lead us to consider local microcircuits with spatial extension of order $v(N)/N$, which tends to zero in the limit $N\to\infty$. At this scale, the fluctuations related to the microcircuit vanish, which allowed identifying the large $N$ limit process. However, at the scale of one neuron (or considering, similarly to~\cite{stannat-lucon:13}, a field of size $N$, i.e. typical distances between neurons of order $1$), the microcircuits may induce more complex phenomena in which fluctuations become prominent. This interesting problem remains largely open and cannot be addressed with the techniques presented in the manuscript.}

The developments presented in this article also go way beyond what was done in the domain of mean-field analysis of large spatially extended systems. In that domain, probably the two most relevant contributions to date are~\cite{stannat-lucon:13} and~\cite{touboulNeuralfields:11}. In~\cite{touboulNeuralfields:11}, a relatively sketchy model of neural field was proposed, in which the system was fully connected and neurons gathered at discrete space location that eventually filled the neural field. The model presented here is considerably more relevant from the biological viewpoint, and necessitated to deeply modify the proofs proposed in that manuscript. In particular, the connectivity patterns are now randomized, and the proof is now made independent of results arising in finite-populations networks. Moreover, the two main contributions of the article, namely the singular coupling, was absent of the above cited manuscript. Such coupling was discussed in~\cite{stannat-lucon:13}, where the authors consider the case of network with nearest-neighbors topology (only a local micro-circuit) in which neurons connect to a non-trivial proportion of neurons $P=cN$. There is a substantial difficulty in considering only very local micro-circuits connectivity and sparse macro-circuits. Here, we solved this problem and framed it in a more general setting with multiscale coupling. 

The proof presented in the present manuscript is relatively general. In particular, it can be extended to models with non locally Lipchitz continuous dynamics (as is the case of the classical Fitzhugh-Nagumo model~\cite{fitzhugh:55}) as was presented in~\cite{touboulNeuralfields:11}, or to networks with multiple layers. The results enjoy a relatively broad universality. Indeed, we observe that the limit obtained is independent of the choice of the size of the micro-circuit and sparsity of the macro-circuit (as long as proper scaling is considered). This property shows that the limit is universal: for any choice of function $v(N)$ and $\beta(N)$, the macroscopic limit of our networks are identical. An interesting question is then what would be an optimal choice of functions $v(N)$ and $\beta(N)$ so that the convergence is the fastest. The speed of convergence towards the mean-field equation is hence governed by three quantities: 
			\begin{itemize}
				\item the term $\left(\frac{v(N)}{N}\right)^{\frac 1 {2d}}$ controls the regularity of the law solution of the mean-field equation with respect to space. The larger $v(N)$, the wider the micro-circuit, and therefore the slowest the local convergence. 
				\item  the term $\frac 1 {\sqrt{v(N)}}$ controls the speed of averaging at the micro-circuit scale, which decreases with the size of the micro-circuit $v(N)$. 
				\item The term $\frac 1 {\sqrt{N\beta(N)}}$ controls the speed of the averaging at the macro-circuit scale. This term is of course the smallest when $\beta(N)$ is large. In the biological system under consideration, there is nevertheless an energetic cost to increasing the connectivity level. 
			\end{itemize}
			The two first term corresponding to the micro-circuit convergence properties can give an information on the order of the optimal micro-circuit size. Minima are obtained when $v(N)$ is of order $N^{1/(d+1)}$, e.g. $\sqrt{N}$ in dimension $1$. Other choices may be analyze to optimize other criteria such that information capacity vs energetic considerations, anatomical constraints, size of clusters sharing resources, \ldots. 
			
			Eventually, this result has also implications in neuroscience modeling. In this domain, authors widely use the so-called Wilson-Cowan neural field model (see~\cite{bressloff:12} for a review). This model is given by non-local differential equations of type:
			\[\partial_t u(r,t)=-u(r,t)+\int_{\Gamma} J(r,r') S(u(r',t-\tau(r,r')))\,dr'\]
			where $u$ represents the mean firing-rate of neurons and $S$ corresponds to a sigmoidal function. This type of equations is similar to those obtained in the analysis of fully connected neural fields, as shown in~\cite{touboulNeuralfields:11,touboulNeuralFieldsDynamics:11}, when considering a discrete Wilson-Cowan type of dynamics for the underlying network, i.e. a case where $f(r,t,x)=-x/\theta(r)+I(r,t)$ and $b(x,y)=S(y)$ for $S$ a smooth sigmoidal function. In this case, we showed~\cite{touboulNeuralFieldsDynamics:11} that the solutions were attracted by Gaussian spatially chaotic processes with mean $\mu(r,t)$ and standard deviation $v(r,t)$ satisfying the integro-differential equations:
			\[\begin{cases}
				\partial_t\mu(r,t)=-\frac{\mu(r,t)}{\theta(r)}+I(r,t)+\int_{\Gamma} J(r,r')F(\mu(r',t-\tau(r,r')),v(r',t-\tau(r,r')))\,dr'\\
				\partial_t v(r,t)=-2\frac{v(r,t)}{\theta(r)}+\sigma(r)^2
			\end{cases}\]
			where $F(x,y)=\int_{\R} S(u)e^{-(u-x)^2}{2y}du/\sqrt{2\pi y}$. These are compatible with the neural field equations. However, these actually appear to overlook the complex connectivity pattern, and in particular neglect the additional local averaging term that we found here using rigorous probabilistic methods. Taking into account local microcircuitry would actually yield an additional term in the equation on $\mu(r,t)$:
			\[\partial_t\mu(r,t)=-\frac{\mu(r,t)}{\theta(r)}+I(r,t)+\bar{J}F(\mu(r,t-\tau_s),v(r,t-\tau_s))+\int_{\Gamma} J(r,r')F(\mu(r',t-\tau(r,r')),v(r',t-\tau(r,r')))\,dr'.\]
The study of these new equations will, with no doubt, present substantial different dynamics, are offer a new neural field model well worth analyzing in order to understand the qualitative role of local microcircuits on the dynamics. 

\medskip
{\bf Acknowledgement:} The author deeply acknowledges the help of anonymous reviewers for their important remarks on the manuscript. 
\bibliographystyle{apt}
{ \bibliography{perso,neuromathcomp,odyssee,PhysRev}}
\end{document}